\documentclass[11pt]{article}
\usepackage[a4paper,margin=1in]{geometry}
\usepackage{amsmath,amssymb,amsthm,mathtools}
\usepackage{algorithm}
\usepackage{algpseudocode}
\usepackage{tikz}
\usetikzlibrary{arrows.meta,positioning,calc,fit,backgrounds}
\usepackage{booktabs}
\usepackage{microtype}
\usepackage{enumitem}
\usepackage{hyperref}
\usepackage{xcolor}
\usepackage{caption}
\usepackage{float}
\usepackage{cleveref}
\hypersetup{colorlinks=true,linkcolor=black,citecolor=black,urlcolor=blue}

\newtheorem{theorem}{Theorem}

\newcommand{\Sset}{\mathcal{S}_t}
\newcommand{\Uset}{\mathcal{U}_t}
\newcommand{\Fset}{\mathcal{F}_t}
\newcommand{\Iset}{\mathcal{I}_t}
\newcommand{\PathU}{\mathcal{P}(\mathcal{U}_t)}

\title{\textbf{Dijkstra Is NOT Greedy: A Global-to-Local Proof of Correctness}}
\author{Shengbao Wang\\
\small Hangzhou Normal University, China \\
\small shengbaowang@hznu.edu.cn}
\date{October 13, 2026}

\begin{document}
\maketitle
\vspace{-2.5em}

\begin{abstract}
Dijkstra's algorithm is traditionally classified as a greedy algorithm. This paper argues that this classification misrepresents the mechanism by which the algorithm determines shortest paths. We present a direct global-to-local proof of correctness. At each iteration, the algorithm identifies a shortest path over the entire set of paths whose endpoints remain unsolved. This global minimum is necessarily a shortest path to its own endpoint. More fundamentally, the algorithm uses already determined shortest paths to establish one additional shortest path at each iteration. We identify this progression as dynamic programming. A three-step exclusion proof shows how the full path space is reduced to frontier candidates, how relaxation removes inferior alternatives, and why minimum extraction yields the global minimum. The choice is thus determined by a global comparison, not by an optimistic local greedy choice. We conclude that the algorithm is more accurately characterized by dynamic programming and deterministic path exclusion than by the traditional greedy classification.
\end{abstract}

\section{Introduction}
Dijkstra's shortest-path algorithm, introduced in 1959 \cite{Dijkstra1959}, is a classic of computer science.
It has important practical applications, including route planning in map navigation and routing in communication networks \cite{SedgewickWayne2011}. It has also long been regarded as a standard greedy algorithm. This classification is prominent in classical textbooks \cite{BrassardBratley1988,KleinbergTardos2006,CLRS2022}. Kleinberg and Tardos describe greedy decisions as being made \emph{myopically}, with each step optimizing a local criterion. Cormen et al. characterize Dijkstra as greedy because it selects the unsettled vertex with the smallest tentative distance. Their correctness proof then establishes the finality of the selected vertex's distance. Sedgewick and Wayne present the same minimum-distance and relaxation operations, while Princeton course materials likewise place shortest paths among greedy algorithms \cite{SedgewickWayne2011,WayneCOS423}. Roughgarden offers a particularly clear example: a minimum boundary candidate is selected, and its optimality for the chosen endpoint is subsequently proved \cite{Roughgarden2018}.

These proofs are rigorous. They also employ properties of paths that cross the boundary of the settled set. Their principal conclusion, however, concerns the shortest-path distance of the \emph{selected endpoint}. The choice is stated before its correctness is established. This order supports the customary picture of a local greedy decision followed by a proof of its safety. It does not expose the stronger fact that the selected candidate is the minimum over \emph{all paths to currently unsolved destinations}.

Alternative accounts of Dijkstra's algorithm have also appeared. Building on the optimal-substructure perspective familiar from dynamic programming \cite{Bellman1957,Lawler1976}, Sniedovich emphasizes its dynamic-programming nature, while Misra derives it through fixed-point computation \cite{Sniedovich2006,Misra2001}. Sniedovich nevertheless discusses the algorithm's customary greedy selection rule, and Lew, in examining the connection between the two paradigms, continues to treat Dijkstra as greedy \cite{Lew2006}. These works do not reject the greedy characterization by tracing the selection rule back to a global minimum over \emph{all} paths ending at unsolved vertices, nor do they develop the global-to-local exclusion argument presented here.

This paper takes that global path space as its starting point. We prove that every iteration identifies its absolute shortest path. The shortest path to the selected endpoint then follows directly from this global result. To establish this claim, we give a three-step exclusion proof. The first step restricts the search to the frontier. The second uses optimal substructure and relaxation to retain the best candidate for each frontier vertex. The third determines the minimum over the remaining candidates. These reductions correspond to the operations of Dijkstra's algorithm itself.

The resulting interpretation reverses the familiar explanatory order. The algorithm does not reach a global conclusion by trusting a locally attractive choice. Instead, deterministic exclusions establish the global minimum before its endpoint is declared solved. Each iteration uses previously established shortest paths to obtain a new one. We therefore argue that Dijkstra's algorithm is fundamentally a dynamic-programming algorithm, and that its customary classification as greedy is incorrect. This claim concerns the mathematical basis of its minimum selection, not the validity of earlier correctness proofs.

\section{Preliminaries}
\label{sec:prelim}
Let $G=(V,E,w)$ be a finite directed weighted graph with source $s\in V$. To make the exclusion argument especially transparent, we assume strictly positive edge weights throughout:
\[
  w(e)>0\qquad (e\in E).
\]
Strict positivity makes every proper path extension longer. Dijkstra also works with zero-weight edges: extensions then cannot become shorter, and the first-entry frontier still contains a path attaining the current global minimum. Equal distances may arise along a shortest path, so its vertices need not be settled in strict prefix order. Negative weights do not preserve this argument. We use positive weights to keep the comparisons strict.

By a path we mean a sequence $P=(s=v_0,v_1,\ldots,v_k=v)$, starting at the fixed source $s$, passing through intermediate vertices $v_1,\ldots,v_{k-1}$, and ending at $v$. Throughout, paths are \emph{simple}: no vertex occurs more than once. This entails no loss for shortest paths under positive weights, because removing a cycle from any walk strictly reduces its length. Let $w(P)$ be the sum of the edge weights of $P$, and let $\delta(s,v)$ be the shortest-path distance to a vertex $v$ reachable from $s$. Unreachable vertices are excluded from finite-distance minimizations; if no unsolved vertex is reachable, the procedure terminates.

\subsection{The Three Vertex Sets}
Dijkstra already divided vertices into three classes in Problem 2 of his original paper \cite{Dijkstra1959}. Set A contains vertices with known shortest paths. Set B contains the remaining vertices directly connected to A; the next vertex is chosen from this set. Set C contains all other vertices. We use different symbols but retain this three-region structure.

At iteration $t$, let $\Sset$ be the \emph{solved set}, for which
\[
 d(x)=\delta(s,x)\qquad (x\in\Sset),
\]
and let $\Uset=V\setminus\Sset$ be the \emph{unsolved set}. The initially solved set is $\{s\}$, with $d(s)=0$. Within $\Uset$, define
\[
\Fset=\{f\in\Uset:\exists x\in\Sset\text{ such that }(x,f)\in E\}
\]
as the \emph{frontier} and $\Iset=\Uset\setminus\Fset$ as the \emph{invisible set}. These are disjoint partitions:
\[
 V=\Sset\;\dot\cup\;\Uset,
 \qquad \Uset=\Fset\;\dot\cup\;\Iset.
\]
Here ``invisible'' means only that the vertex is not directly exposed by an outgoing edge from the current solved set. Figure~\ref{fig:spatial} illustrates Dijkstra's original three-region partition in our notation. For directed graphs, the frontier uses edges directed \emph{out of} the solved region.

\begin{figure}[H]
\centering
\begin{tikzpicture}[
    >=Latex,
    vertex/.style={circle,draw,minimum size=7mm,inner sep=1pt},
    every node/.style={font=\small}
]
\draw[thick,rounded corners] (-0.4,-2.75) rectangle (11.1,2.75);
\node[anchor=north west] at (-0.25,2.60) {$V$};

\draw[thick,rounded corners] (0.0,-2.10) rectangle (3.75,2.10);
\node at (1.88,1.78) {$\mathcal{S}_t$ (Solved)};

\draw[thick,rounded corners] (4.15,-2.10) rectangle (10.70,2.10);
\node at (7.43,1.78) {$\mathcal{U}_t$ (Unsolved)};

\draw[dashed,rounded corners] (4.45,-1.52) rectangle (6.40,1.35);
\node at (5.43,1.07) {$\mathcal{F}_t$};

\draw[dashed,rounded corners] (6.80,-1.52) rectangle (10.30,1.35);
\node at (8.55,1.07) {$\mathcal{I}_t$};

\node[vertex] (s) at (0.60,0.05) {$s$};
\node[vertex] (a) at (1.70,0.88) {$a$};
\node[vertex] (b) at (1.75,-0.88) {$b$};
\node[vertex] (x) at (3.00,0.05) {$x$};

\draw[->,thick] (s) -- (a);
\draw[->,thick] (a) -- (x);
\draw[->,thick] (s) -- (b);
\draw[->,thick] (b) -- (x);

\node[vertex] (f1) at (5.05,0.50) {$f_1$};
\node[vertex] (f2) at (5.70,-0.70) {$f_2$};
\draw[->,very thick] (x) -- (f1);
\draw[->,very thick] (b) -- (f2);

\node[vertex] (u1) at (7.70,0.50) {$u_1$};
\node[vertex] (u2) at (8.75,-0.45) {$u_2$};
\node[vertex] (u3) at (9.45,0.62) {$u_3$};
\draw[->,dashed] (f1) -- (u1);
\draw[->,dashed] (u1) -- (u3);
\draw[->,dashed] (f2) -- (u2);
\draw[->,dashed] (u2) -- (u3);

\node[font=\scriptsize,fill=white,inner sep=1pt] at (6.38,0.72) {$\times$ longer path};

\draw[densely dotted,thick] (3.95,-2.28) -- (3.95,2.28);
\node[rotate=90,font=\scriptsize] at (3.89,0) {first-entry boundary};

\end{tikzpicture}
\caption{Dijkstra's three-way partition at iteration $t$, in our notation. Every path from $s$ to an unsolved vertex first enters the frontier $\mathcal{F}_t$. With positive edge weights, any further extension is longer than this first-entry prefix.}
\label{fig:spatial}
\end{figure}
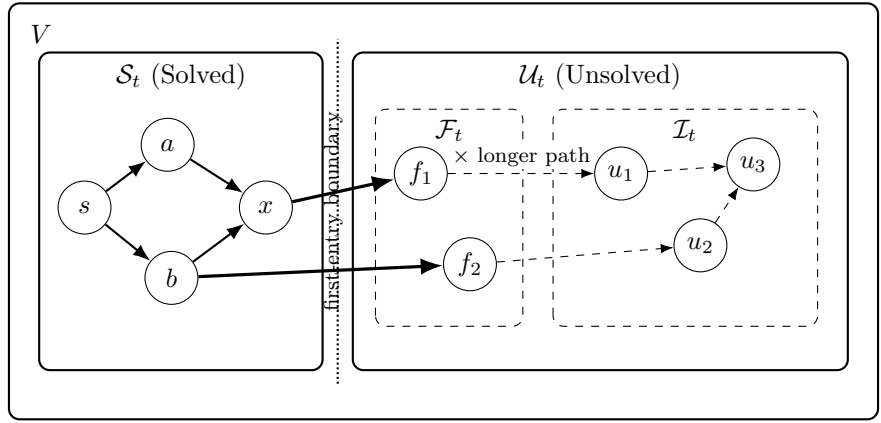

\subsection{A Motivating First Step}
The source $s$ is trivially solved, since the length-zero path to itself has distance $\delta(s,s)=0$; this is an initial fact, not a search step. Dijkstra's first actual step is to solve the shortest-path problem for some vertex other than $s$. From the algorithm's minimum-label rule, it is easy to see that this vertex is the head of a minimum-weight edge from $s$ to a different vertex. If another vertex is reachable, let
\[
  m=\min_{\substack{(s,v)\in E\\v\ne s}}w(s,v).
\]
Each such one-edge path has length $m$. Any simple path from $s$ to another vertex with at least two edges starts with an edge of weight at least $m$ and then adds positive-weight edges. Thus a minimum-weight outgoing edge is not merely shortest to its own head: it is an absolute shortest path among all paths from $s$ to other vertices.

The reason is elementary: extending a path makes it longer. This first step suggests the question at the heart of the paper: \emph{Does every later iteration likewise identify an absolute shortest path among all paths to the vertices still unsolved?} The following global-to-local proof answers yes.

\section{A Global-to-Local Proof by Three-Step Exclusion}
\label{sec:principle}
\label{sec:proof}
\subsection{The Global Path Minimum}
For each $u\in\Uset$, let $\mathcal P(s,u)$ be the set of paths from $s$ to $u$. Define the global search space at iteration $t$ by
\[
  \PathU=\bigcup_{u\in\Uset}\mathcal P(s,u).
\]
The objective of the iteration is to identify
\[
  P_t^*\in\arg\min_{P\in\PathU}w(P).
\]
The minimum is taken over \emph{all paths ending at unsolved vertices}, not over a collection of shortest paths that are assumed to have already been found.

\paragraph{From the global minimum to a local shortest path.}
Suppose $P_t^*$ ends at $u^*$. Since
\[
  \mathcal P(s,u^*)\subseteq\PathU,
\]
every path $Q$ from $s$ to $u^*$ satisfies
\[
  w(P_t^*)\le w(Q).
\]
Thus a path that minimizes the global search space necessarily minimizes the path set for its own endpoint. \emph{One global shortest path therefore settles one additional vertex.} The remaining task is to show how Dijkstra determines this global minimum without enumerating all paths.

\paragraph{The exclusion principle.}
To find a shortest path, one may exclude any path strictly longer than another admissible candidate. If $P_1$ is admissible and $w(P_1)<w(P_2)$, then $P_2$ cannot attain the current minimum. Shortest paths need not be unique, even with positive weights. Throughout, ``a global shortest path'' means any one minimizer; if several optimal candidates have equal length, whether they end at the same vertex or at different vertices, any one of them may be chosen without affecting the shortest distance.

This exclusion refers to the \emph{current iteration}. A path extended beyond the first unsolved vertex may be irrelevant to the present global minimum, but its later edges can become relevant after that vertex is settled. The following proof uses this distinction to reduce the current search space without losing any future shortest-path solution.

\subsection{The Three-Step Exclusion Proof}
Fix iteration $t$ with a nonempty reachable frontier. Distances $d(x)=\delta(s,x)$ for all $x\in\Sset$ are already finalized. We prove that Dijkstra's next extraction realizes the global minimum in $\PathU$.

\paragraph{Step 1: Frontier exclusion.}
Take an arbitrary path
\[
  P=(s=v_0,v_1,\ldots,v_k=u)\in\PathU.
\]
Let $v_j$ be its first unsolved vertex. Since $v_0=s\in\Sset$, its predecessor $v_{j-1}$ belongs to $\Sset$. The edge $(v_{j-1},v_j)\in E$ therefore places $v_j$ in the frontier $\Fset$ by definition. Its prefix $P^{\mathrm{first}}=(v_0,\ldots,v_j)$ also belongs to $\PathU$. If $P$ continues past $v_j$, it must acquire at least one more positive-weight edge, so
\[
  w(P^{\mathrm{first}})<w(P).
\]
Thus no path containing more than one unsolved vertex can be the current global minimum: a shorter path already ends at its first frontier vertex. What remains relevant at this iteration is the form
\[
  s\rightsquigarrow x\to f,
  \qquad x\in\Sset,\quad f\in\Fset,
\]
with every internal vertex in the solved set and exactly one unsolved vertex, its endpoint $f$. This is the mathematical reason why it suffices to consider frontier vertices rather than paths going further into the unsolved region. Here exclusion is relative to the shorter first-entry prefix: a postponed path need not be longer than every other surviving frontier candidate.

\paragraph{Step 2: Relaxation exclusion.}
Fix $f\in\Fset$. Every candidate surviving Step 1 is of the form
\[
  P=Q(s,x)\circ(x,f),\qquad x\in\Sset.
\]
First consider the prefix to $x$. Since $x$ has already been solved, an arbitrary prefix $Q(s,x)$ cannot improve on $d(x)=\delta(s,x)$. If it is strictly longer, replacing it by a finalized shortest prefix produces a strictly shorter candidate to $f$. We therefore need only the candidate value
\[
  d(x)+w(x,f)
\]
for each boundary edge $(x,f)$. This is the optimal-substructure reason why relaxation extends finalized shortest distances, rather than enumerating earlier paths.

Several solved predecessors may lead to the same frontier vertex $f$. The shortest candidate among them is
\[
  D_t(f)=
  \min_{\substack{x\in\Sset\\(x,f)\in E}}
  \{d(x)+w(x,f)\}.
  \label{eq:frontier}
\]
All longer candidates to this same endpoint can now be excluded; if several have the same minimum length, retaining one is sufficient. Thus Step 2 comprises both optimal-prefix replacement and same-endpoint minimization. Both reductions are reflected in the relaxation operation. After Step 2, the current search is represented by the frontier values
\[
  \{D_t(f):f\in\Fset\}.
\]

\paragraph{Step 3: Global minimum extraction.}
Choose
\[
 f^*\in\arg\min_{f\in\Fset}D_t(f),
 \qquad D^*=D_t(f^*).
\]
We show that this is the minimum not only among frontier representatives but among \emph{all} paths in $\PathU$.

Take any $P\in\PathU$. Let $f$ be its first unsolved vertex, $x$ its predecessor, and $L_x$ the length of the prefix from $s$ to $x$. As $x\in\Sset$, $d(x)\le L_x$. Since $D_t(f)$ is the least candidate over all solved-to-$f$ boundary edges,
\[
 D_t(f)\le d(x)+w(x,f)\le L_x+w(x,f).
\]
The remainder of $P$ after $f$ has nonnegative length, so $w(P)\ge L_x+w(x,f)$. By the choice of $f^*$, $D_t(f)\ge D_t(f^*)$. Combining gives
\[
 \begin{aligned}
 w(P)&\ge L_x+w(x,f)\\
     &\ge d(x)+w(x,f)\\
     &\ge D_t(f)\\
     &\ge D_t(f^*)=D^*.
 \end{aligned}
\]
A path attaining $D^*$ exists: choose a finalized shortest prefix to the corresponding solved predecessor and append its boundary edge. Such a prefix may be chosen entirely within the solved set (by induction on earlier extractions), so its extension to an unsolved endpoint is simple. Therefore
\[
  D^*=\min_{P\in\PathU} w(P)
     =\min_{f\in\Fset}D_t(f).
\]
This equality also shows why the eventual minimum-label operation implements a global path minimum: Steps 1 and 2 have already reduced the path space without changing its least value.

\begin{theorem}[Correctness of one iteration]\label{thm:iteration}
If all edge weights are positive and at least one unsolved vertex is reachable, then any $f^*$ minimizing $D_t(f)$ over $\Fset$ satisfies
\[
  D_t(f^*)=\delta(s,f^*).
\]
Hence its shortest-path distance can be finalized.
\end{theorem}
\begin{proof}
By the preceding inequalities, a path attaining $D_t(f^*)$ is globally shortest in $\PathU$. In particular, it is shortest among the paths to its own endpoint $f^*$, which form a subset of $\PathU$.
\end{proof}

The solved state then grows and the unsolved endpoint set shrinks:
\[
 \mathcal S_{t+1}=\Sset\cup\{f^*\},\qquad
 \mathcal U_{t+1}=\Uset\setminus\{f^*\}.
\]
The same reasoning applies at the next iteration.

\subsection{Correspondence with the standard algorithm}
\label{sec:correspondence}
The three exclusions are not additional operations. They account for the operations already performed by the standard algorithm, as summarized in Table~\ref{tab:correspondence}. Frontier restriction implements Step 1; relaxation maintains the reduced candidate values of Step 2; and \textsc{Extract-Min} implements Step 3.

\begin{table}[htbp]
\centering
\caption{Exclusion steps and corresponding operations in Dijkstra's algorithm.}
\label{tab:correspondence}
\small
\begin{tabular}{@{}p{0.24\linewidth}p{0.29\linewidth}p{0.39\linewidth}@{}}
\toprule
Exclusion step & Standard operation & Mathematical role \\
\midrule
Step 1: Frontier exclusion & Frontier restriction & Omit paths that continue beyond their first unsolved vertex. \\
Step 2: Relaxation exclusion & Relaxation of tentative distances & Use optimal solved prefixes and retain the best candidate for each frontier endpoint. \\
Step 3: Global minimum & \textsc{Extract-Min} & Select the global minimum among the surviving candidates. \\
\bottomrule
\end{tabular}
\end{table}

To make the correspondence between the three-step exclusion proof and the original algorithm more explicit, Algorithm~\ref{alg:standard} presents a compact version of Dijkstra's procedure, with the relevant operations annotated. It starts with the source $s$ in the solved set and immediately relaxes its outgoing edges. The first minimum extraction therefore uses frontier values that have already been computed. In later iterations, the relaxations following an extraction update the frontier values for the next extraction. Hence, although the loop places \textsc{Extract-Min} before relaxation, the state used by each extraction already reflects Step~2. The two edge scans make Step~1 explicit: only edges from the solved region to unsettled heads are examined, so only frontier vertices can receive a finite candidate.

\begin{algorithm}[htbp]
\caption{Dijkstra's algorithm with the source settled initially}\label{alg:standard}
\begin{algorithmic}[1]
\State $d(s)\gets 0$; $d(v)\gets\infty$ for $v\ne s$
\State $\mathcal S\gets\{s\}$; $Q\gets V\setminus\{s\}$ (queue keyed by $d$)
\For{each $(s,v)\in E$ with $v\in Q$} \Comment{Step 1: frontier heads only}
  \State $d(v)\gets\min\{d(v),w(s,v)\}$; update $v$'s key in $Q$ \Comment{Step 2: initial relaxation}
\EndFor
\While{$Q\ne\varnothing$}
  \State $u\gets\operatorname{ExtractMin}(Q)$ \Comment{Step 3: frontier minimum}
  \If{$d(u)=\infty$} \State \textbf{break} \EndIf
  \State $\mathcal S\gets\mathcal S\cup\{u\}$
  \For{each $(u,v)\in E$ with $v\in Q$} \Comment{Step 1: frontier heads only}
    \If{$d(u)+w(u,v)<d(v)$}
      \State $d(v)\gets d(u)+w(u,v)$; update $v$'s key in $Q$ \Comment{Step 2: relaxation}
    \EndIf
  \EndFor
\EndWhile
\end{algorithmic}
\end{algorithm}

Once a vertex $u$ is settled, its distance $d(u)=\delta(s,u)$ is final, so every outgoing edge $(u,v)$ supplies its best possible extension value $d(u)+w(u,v)$ when scanned. In a standard adjacency-list procedure, each outgoing edge of a reachable vertex is examined at most once, and a relaxation is attempted only if its endpoint remains unsettled; a frontier vertex's tentative label may nevertheless decrease several times. In contrast, standard Bellman--Ford repeatedly scans edges because the distance of an edge's tail may still improve. This difference helps explain Dijkstra's efficiency.

Dijkstra's original account already uses three vertex classes and a sequence of candidate updates followed by minimum selection \cite{Dijkstra1959}. The present proof identifies the global path-minimization argument behind that established process. It is a proof of the original algorithm, not a modified algorithm.

\section{Discussion: Dynamic Programming Rather Than Greedy Selection}
\label{sec:discussion}
\subsection{The dynamic-programming recurrence}
Equation~\eqref{eq:frontier} computes each frontier candidate from finalized shortest distances. Once the global frontier minimum is settled, that distance becomes available for subsequent iterations. Hence shortest paths are built incrementally from previously solved subproblems. If $L_t=\min_{P\in\mathcal P(\mathcal U_t)}w(P)$, then
$\mathcal P(\mathcal U_{t+1})\subseteq\mathcal P(\mathcal U_t)$ gives $L_{t+1}\ge L_t$: shortest paths are settled in nondecreasing order of length, even though distinct endpoints may have equal distances. Under positive weights, every proper source-prefix of a shortest path is strictly shorter and is therefore settled earlier. Discovery order need not agree with this order: a vertex may enter the frontier earlier but be settled later.

This recurrence may be compared with Fibonacci computation,
\[
 F_0=0,\qquad F_1=1,\qquad F_n=F_{n-1}+F_{n-2}.
\]
Both compute new exact results from solved subproblems. The next Fibonacci index is predetermined; Dijkstra determines the next endpoint from the current frontier. The common feature is reliance on established subproblem solutions, not on a guessed next result.

\subsection{Why the greedy classification is incorrect}
A greedy rule, in its basic sense, commits to the best immediately available local choice without accounting for the entire remaining problem; Kleinberg and Tardos describe this viewpoint as \emph{myopic} \cite{KleinbergTardos2006}. Dijkstra's \textsc{Extract-Min} looks exactly like such a rule in pseudocode. The proof above, however, establishes the opposite logical direction: the minimum extracted from frontier labels has already been shown to equal the minimum over \emph{all} paths ending at unsolved vertices. The apparent local choice is a compact implementation of that global minimization, obtained through deterministic exclusions rather than a myopic decision. Together with the recurrence on solved shortest paths, this is why we regard Dijkstra as dynamic programming, not as a greedy algorithm.

\section{Conclusion}
We have given a global-to-local proof of Dijkstra's algorithm. Three deterministic exclusions reduce the paths to unsolved destinations to frontier candidates, then to one optimal candidate per frontier vertex, and finally to the global shortest path. Its endpoint's shortest-path optimality follows directly from that global minimum. The same argument also clarifies the role of positive edge weights, the order of finalized paths, and why each outgoing edge need only be processed once.

Each iteration uses previously established shortest paths to determine one new shortest path. This is a dynamic-programming process, not a speculative local greedy choice. The result places Dijkstra alongside Bellman--Ford and Floyd--Warshall, which are already widely understood as dynamic-programming shortest-path algorithms \cite{CLRS2022}. It thus provides a more unified conceptual classification of these major classical shortest-path algorithms.

\section*{Acknowledgments}
The author thanks Zimo Wang for helpful discussions that made the central ideas of this work clearer. The author also acknowledges the assistance of AI large language models, specifically Gemini and ChatGPT, in drafting, translation, and editorial refinement of the manuscript.

\end{document}